\documentclass[11pt]{article}

\usepackage{fullpage}

\usepackage{dsfont}
\usepackage{latexsym}
\usepackage{amsmath,amssymb,amsfonts}
\usepackage{amsthm}
\usepackage{mathtools,mathrsfs}
\usepackage{dsfont}

\usepackage{aliascnt}
\usepackage{graphicx,color}
\usepackage[dvipsnames]{xcolor}
\usepackage{multirow,enumerate}
\usepackage{cite,url}
\usepackage{float,xfrac}

\usepackage{algorithmic}
\usepackage{algorithm}

\usepackage{slashbox}

\newtheorem{theorem}{Theorem}[section]
\newtheorem{lemma}[theorem]{Lemma}

\newtheorem{claim}[theorem]{Claim}

\newtheorem{corollary}[theorem]{Corollary}

\newtheorem{definition}{Definition}[section]

\newcommand{\AutoAdjust}[3]{\mathchoice{ \left #1 #2  \right #3}{#1 #2 #3}{#1 #2 #3}{#1 #2 #3} }
\newcommand{\Xcomment}[1]{{}}

\newcommand{\InParentheses}[1]{\AutoAdjust{(}{#1}{)}}
\newcommand{\InBrackets}[1]{\AutoAdjust{[}{#1}{]}}
\newcommand{\Ex}[2][]{\operatorname{\mathbf E}_{#1}\InBrackets{#2}}
\newcommand{\Exlong}[2][]{\operatornamewithlimits{\mathbf E}\limits_{#1}\InBrackets{#2}}
\newcommand{\Prx}[2][]{\operatorname{\mathbf{Pr}}_{#1}\InBrackets{#2}}
\def\prob{\Prx}
\def\expect{\Ex}

\newcommand{\vect}[1]{\ensuremath{\mathbf{#1}}}

\newcommand{\R}{\mathbb{R}}

\newcommand{\auction}{\mathcal{A}}
\newcommand{\nbidder}{n}

\newcommand{\val}{v}
\newcommand{\vali}[1][i]{\val_{#1}}
\newcommand{\vals}{\vect{\val}}

\newcommand{\hval}{\widetilde{v}}
\newcommand{\hvali}[1][i]{\hval_{#1}}
\newcommand{\hvals}{\widetilde{\vect{\val}}}

\newcommand{\dist}{\mathcal{D}}

\newcommand{\F}{{\mathcal{F}}}
\newcommand{\FT}{{\F^{(2)}}}
\newcommand{\FL}{{\F^{(2,\ell)}}}
\newcommand{\EF}{{\text{EFO}}}
\newcommand{\EFL}{{\text{EFO}^{(2)}}}

\newcommand{\weight}{w}

\newcommand{\bn}{\mathbf{B}}

\newcommand{\bid}{b}
\newcommand{\bidi}[1][i]{{\bid_{#1}}}
\newcommand{\bids}{\vect{\bid}}
\newcommand{\bidsmi}[1][i]{\bids_{\text{-}#1}}

\newcommand\maxV{\mbox{{\sc maxV}}}

\newcommand{\Moff}[1][]{\mathcal{M}_{\text{off}}^{#1}}
\newcommand{\Mon}[1][]{\mathcal{M}_{\text{on}}^{#1}}

\begin{document}

\title{Competitive analysis via benchmark decomposition\\[.2in]}

\author{Ning Chen\thanks{Nanyang Technological University, Singapore. Email: \texttt{ningc@ntu.edu.sg}}
\and
Nick Gravin\thanks{Microsoft Research. Email: \texttt{ngravin@microsoft.com}}
\and
Pinyan Lu\thanks{Microsoft Research. Email: \texttt{pinyanl@microsoft.com}}
}

\date{}

\maketitle

\begin{abstract}


We propose a uniform approach for the design and analysis of prior-free competitive auctions and online auctions. Our philosophy is to view the benchmark function as a variable parameter of the model and study a broad class of functions instead of a individual target benchmark. We consider a multitude of well-studied auction settings, and improve upon a few previous results. 
\begin{itemize}
\item Multi-unit auctions. Given a $\beta$-competitive unlimited supply auction, the best previously known multi-unit auction is $2\beta$-competitive.
We design a $(1+\beta)$-competitive auction reducing the ratio from $4.84$ to $3.24$. These results carry over to matroid and position auctions.

\item General downward-closed environments. We design a $6.5$-competitive auction improving upon the ratio of $7.5$.
Our auction is noticeably simpler than the previous best one.

\item Unlimited supply online auctions. Our analysis yields an auction with a competitive ratio of $4.12$, which significantly
narrows the margin of $[4, 4.84]$ previously known for this problem.
\end{itemize}

A particularly important tool in our analysis is a simple decomposition lemma, which allows us to bound the competitive ratio against a sum of
benchmark functions. We use this lemma in a ``divide and conquer'' fashion by dividing the target benchmark into the sum of simpler functions.

\end{abstract}

\setcounter{page}{0}\thispagestyle{empty}
\newpage

\section{Introduction}


%

A common theme in the design and analysis of online algorithms and prior-free auctions is the competitive framework. In both cases, an online algorithm, which
has to make irrevocable decisions online, or a truthful mechanism, whose outcomes must be aligned with bidders' incentives, are competing against a benchmark
corresponding to a desirable outcome. In each specific problem there usually is a unique well-motivated benchmark (rarely a few) with a small gap to its
handicapped competitors. In this sense, a benchmark is treated as an alternative mechanism or an algorithm that does not need to take into account incentives, or
the online nature of a problem. On the other hand, a benchmark is just a mapping form the space of all possible inputs to real values representing the desired level of efficiency. From this point of view, one may fix any reasonable function as a benchmark and see how well a design can fit it. We may even ask this question for a {\em class of benchmark functions}.

This approach may seem more demanding, however can be helpful in the original problem. In particular, some classes of benchmark functions naturally admit well suited designs. Even if these functions are not close enough to the benchmark of interest, we still may use them as building blocks. In particular, we might be able to decompose the original inconvenient benchmark into the sum of simpler ones, each with a good performance guarantee. On the other note, it is often the case that an online or auction problem admits a reduction to a simpler one. It can be only easier to streamline the problem if we have a better understanding of the simpler setting for a variety of benchmarks. In this paper we demonstrate how these ideas apply to a multitude of auction problems including certain online instances.

%
%
%
%
%
%
%
%

We will be looking at single-parameter environments, where an auctioneer sells an abstract service to $\nbidder$ potential customers bidding in the auction. The auctioneer has a feasibility constraint on which sets of agents can be served simultaneously. Each bidder $i$ values the service at a single privately known value $\vali$. We consider a single-round auction, where each bidder submits a sealed bid $\bidi$. After soliciting the bids auctioneer decides on whether each bidder $i$ receives the service and the amount that $i$ pays. If bidder $i$ gets served, his {\em utility} is the difference between his value $\vali$ and his payment; otherwise, the bidder pays $0$ and his utility is $0$. The auctioneer's {\em revenue} is the total payment of the bidders. We consider the following well-studied single-parameter settings.

\begin{itemize}
  \item Digital goods. The auctioneer sells unlimited number of copies of a single item (e.g., digital goods). Thus, any set of winners is feasible.
  \item Limited supply (a.k.a. multi-unit) auctions. There are $\ell$ copies of the item in total, and thus, at most $\ell$ bidders can be served.
	      There is a ratio-preserving reduction \cite{HartlineY11} to limited supply auctions from position and matroid environments, which have
				applications in FCC spectrum auction and sponsored search. All our results for multi-unit environment carry over to these settings.
  \item General downward-closed permutation environments. Set system of feasible sets is downward-closed if any subset of a feasible set is also feasible.
	      The auctioneer has a probabilistic feasibility constraint, i.e., given by a probability distribution
				over feasible sets. Bidders are assumed to be symmetric, i.e., bidders' values are revealed in a random order. This is a generalization of
				multi-unit and position auctions, as well as matching environments, where each feasible set represents vertices on the one side of a bipartite matching.
  \item Online auctions. The bidders arrive online one by one and the auctioneer makes an irrevocable decision (whether the bidder receives a service or not
	      and at what price) immediately as each new bidder arrives. Online auctions capture important scenarios happening continuously in the Internet, where
				bidders may appear at any time and want to receive service right away.
\end{itemize}


The scope of our work is {\em prior-free auction design} as opposed to the more traditional in economics, Bayesian framework, where buyers' preferences are assumed to come from a given distribution and performance is measured in expectation over the prior distribution. In prior-free settings we look for auctions that perform well (with respect to a benchmark) without any prior knowledge of the bidders' preferences. These auctions are robust to changes of the prior distribution and often give a good approximation to the auctions that are specifically tailored to a distribution.


\subsection{Benchmarks and Competitive Analysis}

We assume that all bidders are selfish and aim to maximize their own utility. We say that an auction is {\em truthful} or incentive compatible if it is a dominant strategy for every bidder $i$ to bid his private value, i.e., $\bidi=\vali$, no matter how other bidders behave. A randomized auction is (universally) truthful if it is given by a distribution over deterministic truthful auctions.

The objective is to design auctions that maximize revenue of the auctioneer. To evaluate the performance of an auction, we need to define a benchmark function $f: \R^\nbidder\rightarrow\R$,
where $f(\bids)$ measures our target revenue for the bid vector $\bids\in \R^\nbidder$.
Given a benchmark function $f(\cdot)$, we say that an auction $\auction$ has a {\em competitive ratio} of $\lambda$ with respect to $f(\cdot)$ if

\[
\Ex{\auction(\bids)}\ge \frac{f(\bids)}{\lambda}, \quad \forall \bids=(\bidi[1],\ldots,\bidi[\nbidder])\in \R^\nbidder
\]

where $\Ex{\auction(\bids)}$ is the expected revenue of auction $\auction$ on the bid vector $\bids$.

A benchmark function should be, on the one hand, economically meaningful in providing a revenue target, and on the other hand not too ambitious so that a truthful auction may have small competitive ratio against the benchmark. For the auction with unlimited supply, the most well-studied benchmark is $\FT(\bids)= \max_{2\le k\le\nbidder} \ k\cdot \bidi[(k)]$, where bids are reordered so that $\bidi[(1)]\ge \bidi[(2)]\ge \cdots \ge \bidi[(\nbidder)]$. That is, $\FT$ gives the largest possible revenue of a fixed price sale given that there are at least two buyers. The reason for having at least two winners is that otherwise, all but one bidders may have $0$ value and then no truthful auction can be competitive against a single bidder with arbitrary large private value. Another meaningful benchmark is $\maxV(\bids)= \max_{1\le k<\nbidder} \ k\cdot \bidi[(k+1)]$. We note that $k\cdot \bidi[(k+1)]$ is the revenue of the Vickrey auction selling $k$ units. Hence, $\maxV$ is the maximum revenue of a $k$-unit Vickrey auction for all possible values of the supply $k$.

For the $\ell$-unit auctions, one can naturally extend the definition of $\F^{(2)}$ to
\[\FL(\bids)=\max_{2\le k\le\ell} k \cdot \bidi[(k)].\]
This is the largest possible revenue of a fixed price sale given that there are at least two and at most $\ell$ buyers (as there are only $\ell$ copies).

Hartline and Yan~\cite{HartlineY11} gave another interpretation of $\FT$ for the unlimited supply setting. Namely, it is the optimal revenue one can extract  in an envy-free allocation with at least two winners. The definition was extended by \cite{HartlineY11} to more general environments such as limited supply and general downward-closed environments. We denote by $\EFL$ the largest revenue that can be obtained in an envy-free allocation for a slightly modified bid vector $\bids^2=(\bidi[2],\bidi[2],\dots,\bidi[n])$.
Interestingly, although $\EFL$ coincides with $\FT$ in the digital goods setting, $\EFL$ is not the same as $\FL$ for the limited supply case, where the precise formula for $\EFL$ will be given in Section~\ref{sec:limited}.

For the online setting, we focus on the model of Koutsoupias and Pierrakos~\cite{KoutsoupiasP13} of unlimited supply auctions competing against the benchmark $\FT$ and $\maxV$. We assume random arrival order of the bidders, as if bidders arrive in an adversarial order, competitive ratio cannot be a constant~\cite{KoutsoupiasP13}.

%
%
%
%
%
%
%
%
%

\subsection{Results and Techniques}

Our recent work~\cite{CGL14} on digital goods auction proposed a uniform procedure for calculating the {\em optimal} competitive ratio against any monotone benchmark. In particular, it yielded tight competitive ratios against the $\maxV$ and $\FT$ benchmarks. Here, we study the design of competitive auctions in other settings. We summarize previous and our new results in the following table. (Since some of the earlier work used auctions for digital goods as a black-box, we update their bounds accordingly with the new tight bounds of~\cite{CGL14}.)


\vspace{.1in}
\begin{tabular}{|c|c|c|c|c|}
  \hline
   & Limited supply & Downward-closed & Online & Online \\
      & $\EFL$ & $\EFL$ & $\FT$ & $\maxV$ \\
     \hline
  Previous upper bounds  & 4.84~\cite{HartlineY11} & 7.5~\cite{DevanurHH13} & 4.84~\cite{KoutsoupiasP13} & - \\
    \hline
  Our upper bounds & 3.24 & 6.5 & 4.12 & 2.42 \\
    \hline
  Lower bounds & 2.42~\cite{GoldbergHKS04} & 2.42~\cite{GoldbergHKS04} & 4~\cite{KoutsoupiasP13} & 2 \\
  \hline
\end{tabular}

\vspace{.1in}

All bounds in the table are for the worst-case scenarios when the number of bidders $\nbidder$ can be arbitrarily large. Better bounds are know for every fixed number of bidders, although, as $\nbidder$ grows, these bounds quickly converge to the worst-case bounds given in the table.

An important new perspective of~\cite{CGL14} was to view benchmark function as a variable parameter of the model. It should be noted that to a limited degree an earlier work \cite{HartlineM05} also studied a broad class of benchmarks in the digital goods setting.
In this paper, we continue to explore this idea in more general settings. Unlike~\cite{CGL14} we {\em explicitly design} auctions with improved competitive ratios.

The following general yet simple observation appears to be very helpful in our analysis. When we seek for an auction with good performance against a specific benchmark $f(\bids)$, it is often the case that $f(\bids)$ can be decomposed into the sum of two functions $f(\bids)=f_1(\bids)+f_2(\bids)$, such that it is easier to find good competitive auctions separately against $f_1(\cdot)$ and $f_2(\cdot)$. The following lemma gives an upper bound on the competitive ratio against the original benchmark function $f(\cdot)$.

\begin{lemma}[Decomposition lemma]\label{le:benchmark_sum}
Let $\auction_1$ and $\auction_2$ be truthful $\lambda_1$ and $\lambda_2$ competitive auctions against the benchmarks $f_1(\cdot)$ and $f_2(\cdot)$, respectively. Then there is a truthful $\lambda_1+\lambda_2$ competitive auction against the benchmark $f_1(\cdot)+f_2(\cdot)$.
\end{lemma}
\begin{proof}
We construct an auction that runs $\auction_1$ with probability $\frac{\lambda_1}{\lambda_1+\lambda_2}$ and runs $\auction_2$ with probability $\frac{\lambda_2}{\lambda_1+\lambda_2}$. The constructed auction is (universally) truthful by definition. Its performance for any bid vector $\bids$ is at least
\[
\frac{\lambda_1}{\lambda_1+\lambda_2}\auction_1(\bids)+\frac{\lambda_2}{\lambda_1+\lambda_2}\auction_2(\bids)\ge
\frac{\lambda_1}{\lambda_1+\lambda_2}\frac{f_1(\bids)}{\lambda_1}+\frac{\lambda_2}{\lambda_1+\lambda_2}\frac{f_2(\bids)}{\lambda_2}=
\frac{f_1(\bids)+f_2(\bids)}{\lambda_1+\lambda_2}.
\]
\end{proof}



\subsection{Related Work}

The worst-case study of digital goods auctions was initiated by Goldberg et al.~\cite{GoldbergHW01}. The competitive prior-free framework was formulated by Fiat et al.~\cite{FiatGHK02}. Over the past decade a lot of effort has been devoted to improving the analysis and competitive ratios of digital goods auctions, see, e.g., \cite{GoldbergH03}, \cite{GoldbergHKS04}, \cite{HartlineM05}, \cite{FeigeFHK05}, \cite{AlaeiMS09} and \cite{IchibaI10}. In our recent work~\cite{CGL14}, we showed the optimal bound on the competitive ratio for digital goods auctions.

A few other closely related settings have stemmed from the study of digital goods auctions with the most immediate extension being the limited supply auctions also known as multi-unit auctions. Multi-unit environments have been traditionally studied with respect to the $\FL$ benchmark, which allows a straightforward reduction \cite{GoldbergHKSW06} to the unlimited supply case with a specific number of bidders. Thus optimal bounds of \cite{CGL14} carry over to the multi-unit auctions with respect to the $\FL$ benchmark.


The general downward-closed single-parameter environments include, e.g,. matching, matroid, and position auctions have also received considerable attention in recent years. Hartline and Yan~\cite{HartlineY11} characterized the optimal revenue in the envy-free outcomes and proposed
$\EFL$ as a uniform benchmark for all of these environments.
They presented a truthful multi-unit auction with a constant competitive ratio and established a no-loss reduction from position and matroid auctions to a simpler multi-unit setting. Devanur et al.~\cite{DevanurHY12} improved the competitive ratio to $9.6$ and gave a $189$-competitive auction for the more general downward-closed environments. Ha and Hartline~\cite{HaH12a} further improved the competitive ratio to $30.4$ for the downward-closed environments. In an unpublished followup paper~\cite{HaH12b}, the authors presented a 11-competitive auction using elegant combination of biased sampling and profit extraction ideas. The best known ratio is due to Devanur at al.~\cite{DevanurHH13} (official version of ~\cite{HaH12b}) with 
a $7.5$-competitive auction that builds upon the biased sampling approach in a significantly more intricate manner than in~\cite{HaH12b} .

As a multi-parameter extension of the digital goods auctions setting, Gravin and Lu~\cite{GravinL13} studied competitive auction in the presence of positive externalities among the buyers.

Another thread of work considers digital goods auctions is in the online framework. Motivated by internet advertising, Mahdian and Saberi~\cite{MahdianS06} proposed a model where supply is unknown in advance. Devanur and Hartline~\cite{DevanurH09} studied prior-free auctions in this model and by applying random sampling technique derived results in the prior-free setting.
There was substantial interest from machine learning community~\cite{BlumH05}, \cite{BlumKRW04},\cite{BalcanBHM08} in a closely related online pricing problem. However, this work together with an earlier work~\cite{Bar-YossefHW02} on online auctions does not assume random order of arrivals. It also uses machine learning techniques resulting in a worse performance guarantees that depends on $h$, the ratio between the highest and the lowest bid. Lastly, the setting of Koutsoupias and Pierrakos~\cite{KoutsoupiasP13} is closely related to generalized secretary problem (for a survey on secretary problem and online digital goods auction see~\cite{BabaioffIKK08}). They gave a black-box reduction of the online problem to the standard off-line digital goods setting with a factor 2 loss in the competitive ratio.

\section{Limited supply}
\label{sec:limited}
%

It was pointed out in \cite{GoldbergHKSW06} that there is an equivalence between the unlimited supply auction problem for the $\FT$ benchmark and the limited
supply auction problem for the $\FL$ benchmark. Namely, any unlimited supply auction with $\ell$ bidders that is $\beta$-competitive against $\FT$ can be
converted into a $\beta$-competitive $\ell$-unit auction against $\FL$. This equivalence and the tight results of \cite{CGL14} for unlimited supply auctions with $\ell$ bidders against $\FT$ benchmark yield tight results for $\ell$-unit auction against $\FL$
benchmark with the same competitive ratio $\lambda_\ell$.\footnote{$\lambda_\ell = 1-\sum\limits_{i=2}^\ell \frac{i}{i-1} \binom{\ell-1}{i-1}\left(\frac{-1}{\ell}\right)^{i-1},$ which converges to 2.42 when $\ell$ approaches infinity.}

A similar equivalence was established in \cite{CGL14} between $\ell$-unit auctions competing with any benchmark $f(\cdot)$ that depends only on the first $\ell$
highest bids and unlimited supply auctions with $\ell$ bidders. However, beyond such benchmarks \cite{CGL14} does not provide a satisfactory way to compute tight competitive ratios in the limited supply case. For example, economically meaningful benchmark $\EFL$ of Hartline and Yan~\cite{HartlineY11} depends not only on the first $\ell$ highest bids.


%
%

\begin{definition}
For a fixed valuation profile $\vals$, order all valuations by $v_{(1)}\ge v_{(2)} \ge \cdots \ge v_{(n)}$
and let $g(j)=j\cdot\vali[(j)]$ for each $2\le j\le n$. Consider the concave envelope $\widehat{g}(\cdot)$ of the function $g(\cdot)$ on the interval $[2,n]$. Formally, for each $2\le j\le n$
\[
\widehat{g}(j)=\max_{\substack{i,k: \\ 2\le i\le j \le k\le n}} \left[g(i)\cdot\frac{k-j}{k-i}+g(k)\cdot\frac{j-i}{k-i}\right].
\]
For $\ell$-unit auction
$
\EFL(\vals)=\max\limits_{2\le i\le\ell}\widehat{g}(i).
$
\end{definition}

There is only a constant gap between $\EFL$ and $\FL$ benchmarks for $\ell$-unit auctions.

\begin{lemma}\label{le:EFOL}
For any valuation profile $\vals$, \[\FL(\vals)\le \EFL(\vals)\le \FL(\vals)+(\ell-2)\cdot\vali[(\ell+1)].\]
\end{lemma}
\begin{proof}
The first inequality holds because
\[
\FL(\vals)=\max_{2\le i\le\ell}g(i)\le\max_{2\le i\le\ell}\widehat{g}(i)=\EFL(\vals).
\]

We next prove the second inequality. Let the maximum in the definition of $\EFL(\vals)=\max\limits_{2\le j\le \ell}\widehat{g}(j)$ be attained at $j^*$, i.e., $\EFL(\vals)=\widehat{g}(j^*)$. Let $2\le i^*\le j^*\le k^*\le n$ be such that
\[
\widehat{g}(j^*)= g(i^*)\cdot\frac{k^*-j^*}{k^*-i^*}+g(k^*)\cdot\frac{j^*-i^*}{k^*-i^*}.
\]
We observe that $\FL(\vals)\ge i^*\cdot\vali[(i^*)]= g(i^*)$, since $2\le i^*\le j^* \le \ell$.

If $k^*\le\ell$, then $\FL(\vals)\ge g(k^*)$. Moreover, $\FL(\vals)$ is greater than any convex combination of $g(i^*)$ and $g(k^*)$ and we get $\FL(\vals)\ge\EFL(\vals)$ as desired.
On the other hand, if $k^*\ge\ell+1$, we have
\begin{eqnarray*}
\EFL(\val) & = & g(i^*)\cdot\frac{k^*-j^*}{k^*-i^*}+g(k^*)\cdot\frac{j^*-i^*}{k^*-i^*} \\
           & = & g(i^*)\cdot\frac{k^*-j^*}{k^*-i^*}+k^*\cdot\vali[(k^*)]\cdot\frac{j^*-i^*}{k^*-i^*}\\
					 & = & g(i^*)- g(i^*)\cdot\frac{j^*-i^*}{k^*-i^*} + (j^*-i^*)\cdot\vali[(k^*)]\cdot\frac{i^*}{k^*-i^*}+ (j^*-i^*)\cdot\vali[(k^*)]\\
					 & = & g(i^*)- \vali[(i^*)]\cdot\frac{i^*(j^*-i^*)}{k^*-i^*} + \vali[(k^*)]\cdot\frac{i^*(j^*-i^*)}{k^*-i^*}+ (j^*-i^*)\cdot\vali[(k^*)]\\
					 &\le& g(i^*)  + (j^*-i^*)\cdot\vali[(k^*)] \\
                     &\le& \FL(\vals) + (\ell - 2)\cdot\vali[(\ell+1)].	
\end{eqnarray*}
The last inequality holds, because $g(i^*) \le \FL(\vals)$, $j^*-i^* \le \ell -2$, and $\vali[(k^*)] \le \vali[(\ell+1)]$.
\end{proof}

One can further estimate $\FL(\vals)\ge(\ell-1)\cdot\vali[(\ell)]\ge (\ell - 2)\cdot\vali[(\ell+1)]$; this implies a trivial
upper bound of $2\FL(\vals)$ on $\EFL(\vals)$. As $\lambda_{\ell}$ is the exact competitive ratio
against the $\FL$ benchmark, the competitive ratio against $\EFL(\vals)$ lies between
$\lambda_{\ell}$ and $2\lambda_{\ell}$. These two bounds were the best currently known \cite{DevanurHY12}. However, these bounds are not tight. In particular, we can improve on the upper bound.

We decompose the upper bound of Lemma~\ref{le:EFOL} on $\EFL$ into the sum of two benchmarks $f_1(\vals)=\FL(\vals)$ and
$f_2(\vals)=(\ell-2)\cdot\vali[(\ell+1)]$. The competitive ratio against
the first benchmark is $\lambda_{\ell}$. On the other hand, the revenue of VCG mechanism selling
$\ell$ items is $\ell\cdot\vali[(\ell+1)]$, which shows that the competitive ratio against
$f_2(\vals)$ is $\frac{\ell-2}{\ell}$. By combining Lemma~\ref{le:benchmark_sum} and Lemma~\ref{le:EFOL} we obtain the following claim,
which improves the upper bound on the competitive ratio against $\EFL$ to $\lambda_{\ell}+1$.

\begin{theorem}
For multi-unit auctions with $\ell$ units for sale, there is a $\InParentheses{\lambda_{\ell}+\frac{\ell-2}{\ell}}$-competitive
auction against the $\EFL$ benchmark, where $\lambda_{\ell}$ is the optimal competitive ratio of unlimited supply auction with $\ell$ bidders against the $\FT$ benchmark.
\end{theorem}

On the other hand, it can be seen that the lower bound of $\lambda_\ell$ is not tight. In particular, for any truthful auction $\auction$, we can write the following lower bound by employing the equal revenue distribution $\dist^n$.\footnote{In the equal revenue distribution, all values $\vali$ are drawn identically and independently with probability $\Prx{\vali > x} = \frac{1}{x}$ for any $x\ge 1$. A remarkable property of this distribution is that any truthful auction has the same expected revenue $\nbidder$.}
\[
\frac{\Exlong[\vals\sim\dist^n]{\EFL(\vals)}}{\Exlong[\vals\sim\dist^n]{\auction(\vals)}} >
\frac{\Exlong[\vals\sim\dist^n]{\FL(\vals)}}{\Exlong[\vals\sim\dist^n]{\auction(\vals)}}=
\frac{n\cdot\lambda_\ell}{n}=\lambda_\ell,
\]
where we plugged in the first equality the expected value of $\FL$ over equal revenue distribution
$\Ex{\FL(\vals)}=n\cdot\lambda_\ell$. One can compute this expectation using certain properties of equal revenue distribution:

%

\[
\Ex{\FL(\vals)}=\Exlong[\vals\sim\dist^n]{\vali[(\ell+1)]}\cdot\Exlong[\hvals\sim\dist^\ell]{\FT(\hvals)}=
\Ex{\vali[(\ell+1)]}\cdot\ell\cdot\lambda_\ell.
\]
This lower bound for each $\ell$ is strictly greater than $\lambda_\ell$. Although, since $\EFL(\vals)\le\FT(\vals)$, this approach cannot give us a bound better than $\lambda_n$.
This naturally makes us conjecture that the competitive ratio against the benchmark $\EFL(\vals)$ for any supply $\ell$ never exceeds competitive ratio $\lambda=\lim_{n\to\infty}\lambda_n$ of the optimal unlimited supply auction.

\section{Downward-Closed environments}

In this section, we consider general downward-closed permutation environments. We denote by $\EF(\vals)$ the optimal revenue achievable in an envy-free allocation for the vector of values $\vals$. Our benchmark of interest is $\EFL(\vals)=
\EF(\vali[2],\vali[2],\dots,\vali[n])$.
The basic ingredients of our auction are biased random sampling and the profit extraction auction from \cite{HaH12b}. Our auction is slightly different from the one presented in \cite{HaH12b} and has a better competitive ratio of $6.51$ compared to $11$ of \cite{HaH12b}. It is much simpler than another auction with competitive ratio $7.5$ presented in \cite{DevanurHH13}, which has a few more components on the top of random sampling and profit extraction.

The profit extraction (PER) auction 
receives as a parameter a target valuation profile $\hvals$. When $\text{PER}^{\hvals}(\vals)$ is run on the actual valuation profile $\vals$, it is able to
extract revenue greater than or equal to the value of the envy-free benchmark $\EF(\hvals)$ as well as $\EFL(\hvals)$, if profile $\vals$ dominates $\hvals$ ($\vals\succeq\hvals$), i.e.,
$\vali[(i)]\ge\hvali[(i)]$ for every bidder $i$. If $\vals\not\succeq\hvals$, PER rejects all bidders.

\begin{lemma}[Ha and Hartline, 2012]
For any downward-closed permutation environment, there is a truthful
profit extraction auction $\text{PER}^{\hvals}(\vals)$ with a profit of at least $\EFL(\hvals)$, if
$\vali[(i)]\ge\hvali[(i)]$ for each bidder $i$.
\end{lemma}

Our auction is quite simple: with some probability $p$ we run the single-item Vickrey auction; with probability $1-p$ we run the following $\sigma$-biased random-sampling profit-extraction auction (denoted by $\sigma$-BSPE).

\begin{itemize}
\item Divide all bidders into two groups market $M$ and sample $S$: Place the \textit{two highest} bids in $M$. Sample the rest bids independently with
      probability $\sigma<1/2$ in $S$ and with probability $1-\sigma$ in $M$.
\item Let $\hvals=\vals_{S}$. Allocate to the winners of $\text{PER}^{\hvals}(\vals_{_M})$.
\end{itemize}

\begin{theorem}
For any downward-closed permutation environment, $\sigma$-BSPE is a $6.51$-competitive truthful auction
against the envy-free benchmark $\EFL(\vals)$.
\end{theorem}

\begin{proof}
For any random coin flips of $\sigma$-BSPE, the allocation rule of $\text{PER}^{\hvals}$ is monotone. This implies that $\sigma$-BSPE also has a monotone allocation rule. Therefore, since our environment is a single-parameter domain, $\sigma$-BSPE allocation with the threshold payment rule makes the auction universally truthful.

We next estimate the expected revenue of $\sigma$-BSPE. We follow closely the proof strategy described
in \cite{HaH12b}, the main difference being in the way we deal with the benchmark $\EFL(\vals)$.
We note that if $\vals_{M}\succeq\vals_{S}$, then the total sum of the threshold payments of
$\text{PER}^{\hvals}(\vals_{M})$ is at least $\EF(\vals_{S})$; we further observe that the
threshold payments of $\sigma$-BSPE can be only larger than that, as we could only increase
payments of the two highest bidders.

\begin{claim}
The probability that $\vals_{M}\succeq\vals_{S}$ is at least $1-(\frac{\sigma}{1-\sigma})^3$.
\end{claim}
\begin{proof}
Sort all bidders in the original profile $\vals: \vali[(1)]\ge \dots\ge \vali[(n)]$ (without loss of
generality we assume that all inequalities are strict). We simulate our random sampling process
by independently flipping a biased coin for each bidder $(i)$ in this order. Each time we count
the difference between the number of bidders in $M$ and $S$. Note that because we always place the highest two bids in $M$,
after the first two steps the difference becomes two. Note that $\vals_{M}\not\succeq\vals_{S}$ if and only if at some step
$(i)$ this difference becomes negative. We next estimate the probability that this
event never happens.

We consider an infinite random walk on a one-dimensional infinite line;
each time we move to the left with probability $\sigma$ and to the right with probability $1-\sigma$.
It is well known that the probability that such a random walk starting at a point $x$ eventually
makes one step to the left from $x$ is $\frac{\sigma}{1-\sigma}$. As our random walk starts at point
$2$, it would take three such steps to move below $0$. The probability of this event is
$\InParentheses{\frac{\sigma}{1-\sigma}}^3$. Therefore, the probability that this never happens after
$n$ steps is at least $1-\InParentheses{\frac{\sigma}{1-\sigma}}^3$.
\end{proof}

We conclude that the expected revenue of the $\sigma$-BSPE is at least
\begin{align*}
\Ex{\sigma\mbox{-BSPE}} &\geq  \Ex{\EF(\vals_{S}) ~|~ \vals_{M}\succeq\vals_{S} } \cdot \Prx{\vals_{M}\succeq\vals_{S}}\\
& =\Ex{\EF(\vals_{S})}  -  \Ex{\EF(\vals_{S})  ~|~ \vals_{M}\not \succeq\vals_{S} } \cdot \Prx{\vals_{M}\not \succeq\vals_{S}}\\
&\geq  \sigma\cdot \EF(\vals_{-\{1,2\}}) - \EF(\vals_{-\{1,2\}}) \cdot \Prx{\vals_{M}\not \succeq\vals_{S}}\\
&\geq \InParentheses{\sigma - \Big(\frac{\sigma}{1-\sigma}\Big)^3}  \EF(\vals_{-\{1,2\}}),
\end{align*}
where $\vals_{-\{1,2\}}$ is the bid vector without first two highest bids. The maximum of the function $\InParentheses{\sigma - \InParentheses{\frac{\sigma}{1-\sigma}}^3}$ is attained at $\sigma\approx0.29$ with a rough value of $0.22$. Thus, the competitive ratio of $\sigma$-BSPE against the benchmark $\EF(\vals_{-\{1,2\}})$ is $4.51$.

On the other hand, by running the single-item Vickrey auction, we extract revenue of at least
$\frac{1}{2}\cdot\EF((v_2,v_2))$. Note that by subadditivity of $\EF$ (shown in \cite{HartlineY11}) we have $\EF((v_2,v_2))+\EF(\vals_{-\{1,2\}})\ge \EF((v_2,v_2,\vals_{-\{1,2\}}))=\EFL(\vals)$.
Therefore, according to Lemma~\ref{le:benchmark_sum} one can achieve the competitive ratio of $4.51+2=6.51$ against the benchmark $\EF((v_2,v_2))+\EF(\vals_{-\{1,2\}})$. Thus, we obtain a $6.51$-competitive auction against $\EFL(\vals)$, which runs $0.22$-BSPE with probability $4.51/6.51$ and the single-item Vickrey auction with probability $2/6.51$.
\end{proof}

\section{Online Auctions}

Let $\{\Moff[n]\}_{n=2}^{\infty}$ be a sequence of $\beta$-competitive offline digital goods auctions against a benchmark $f(\cdot)$ for each number of bidders
$n$. To simplify notation, we refer $\{\Moff[n]\}_{n=2}^{\infty}$ as $\Moff$ auction omitting the number of bidders when it could be inferred from the context.

\begin{theorem}[Koutsoupias, Pierrakos \cite{KoutsoupiasP13}]
Let $\Moff$ be a $\beta$-competitive offline auctions against the $\FT$ benchmark. The online sampling auction is a $2\beta$-competitive against the $\FT$ benchmark
with bidders arriving in a random order.
\end{theorem}

The online sampling auction by \cite{KoutsoupiasP13} uses a black box reduction from an offline  digital-goods auction $\Moff$ to construct an online competitive auction $\Mon$. Their auction, upon the arrival of each bidder $k$, observes first $k-1$ bids $\bids^{[k-1]}\triangleq (b_1,\ldots,b_{k-1})$ and runs $\Moff[k](\bids^{[k-1]})$ for bidder $k$.

In particular, \cite{KoutsoupiasP13} used the offline auction of \cite{HartlineM05} with a competitive ratio of $3.24$, they obtained an upper bound of $6.48$.
In a recent paper \cite{CGL14}, it was shown that the optimal competitive ratio of $\Moff$ is in fact $2.42$, which gives an upper bound of $4.84$ for the online problem. There is also a lower bound of $4$ in \cite{KoutsoupiasP13} for online auctions with only $2$ bidders.

\begin{corollary}
The optimal competitive ratio of online auctions is between $4$ and $4.84$.
\end{corollary}

We next propose another simple black-box reduction from offline to online auctions with a better competitive guarantee. Any online auction can be thought of as a sequence of offline auctions run for a set of bidders already present at each time. The main idea of our design is to tailor each of our offline auctions to a different from $\FT$ benchmark so that their combination has good performance with respect to $\FT$.

\begin{theorem}
Let $f(\bids)=\max(4\bidi[2],3\bidi[3],4\bidi[4],\dots,k\bidi[k])$ and
 $\Moff$ be a $\beta$-competitive auction against the $f(\cdot)$ benchmark. Then there is a $\beta$-competitive online auction against $\FT$, where bidders
arrive in a random order.
\label{th:plustwo}
\end{theorem}
\begin{proof}
Any truthful online auction $\Mon$ can be viewed as a weighted combination of offline auctions $\{A^{n}\}_{n=2}^{\infty}$ running on $n=1,2\dots$ bidders
\[
\Mon(\bids)=\sum_{n=1}\frac{1}{n}\cdot A^n(\bids^{[n]}).
\]
Indeed, each time when $\Mon$ observes first $n-1$ bids $\bids^{[n-1]}$ and offers a price to bidder $n$, it could have seen any combination of $n-1$ bids among $\bids^{[n]}$ equally likely, since bidders arrive uniformly at random. Therefore, $\Mon$ derives $\frac{1}{n}$
of the revenue of offline auction $A^n(\bids^{[n]})$. We denote by $\Mon[k]$ the online auction run only up to $k$ rounds, i.e.,
\[
\Mon[k](\bids)=\sum_{n=1}^{k}\frac{1}{n}\cdot A^n(\bids^{[n]}).
\]

We are going to construct our online auction $\Mon$ inductively at each time increasing the number of bidders by one. Namely, we assume that for all $n=k-1$ bidders our $\Mon$ auction is $\beta$-competitive. Next we specify an offline auction $A^k$ which together with $\Mon[k-1]$ is $\beta$-competitive for $k$ bidders.

For $n=1$ bidder $\FT$ is $0$, so $\Mon$ is competitive regardless of $A^1$. By induction hypothesis, we know that for any fixed bid vector $\bids^{[k-1]}$,
\[
\beta\cdot \Mon[k-1](\bids^{[k-1]})\ge\FT(\bids^{[k-1]}).
\]
Since the first $k-1$ bids are uniformly selected from $\bids^{[k]}$, we have
\[
\beta\cdot \Mon[k-1](\bids^{[k]}) = \frac{1}{k}\sum_{i=1}^{k}\beta\cdot \Mon[k-1](\bidsmi^{[k]}) \ge\frac{1}{k}\sum_{i=1}^{k}\FT(\bidsmi^{[k]}).
\]

Let us sort the bids in $\bids^{[k]}$ by $\bidi[1]\ge\bidi[2]\ge\dots\ge\bidi[k]$. We compare the revenue of $\Mon[k-1]$ with each $\bidi[\ell]$, for $2\le\ell\le k$. We have $\FT(\bidsmi^{[k]})\ge \ell\cdot\bidi[\ell]$ for every $i>\ell$. If $\ell>2$ and $i\le\ell$, then
$\FT(\bidsmi^{[k]})\ge (\ell-1)\cdot\bidi[\ell]$. Unfortunately, for $\ell=2$, we cannot write
$\FT(\bidsmi[1]^{[k]})\ge \bidi[2]$ or $\FT(\bidsmi[2]^{[k]})\ge \bidi[2]$. Therefore,
\begin{eqnarray*}
\frac{1}{k}\sum_{i=1}^{k}\FT(\bidsmi^{[k]}) &\ge& \max\InParentheses{\frac{2k-4}{k}\bidi[2],
\frac{3k-3}{k}\bidi[3],\dots,\frac{\ell(k-\ell)+(\ell-1)\ell}{k}\bidi[\ell],\dots,\frac{k^2-k}{k}\bidi[k]}\\
       &\ge& \FT(\bids^{[k]})-\frac{1}{k}\max(4\bidi[2],3\bidi[3],4\bidi[4],\dots,k\bidi[k]).
\end{eqnarray*}
We want the offline auction $A^k$ to have good performance against $f(\bids)=\max(4\bidi[2],3\bidi[3],4\bidi[4],\dots,k\bidi[k])$.
We know that there is a $\beta$-competitive auction $\Moff$ with respect to this benchmark  $f(\cdot)$. Thus, there is a $\beta$-competitive auction for $k$ bidders in the online setting. This completes the proof.
\end{proof}

%

Note that $f(\bids)\le\FT(\bids)+2\bidi[2].$ According to Lemma~\ref{le:benchmark_sum}, we can run
a mixture of the optimal auction against $\FT$ and single-item Vickery
auction against $2\bidi[2]$ to achieve a $(\lambda+2)$-competitive auction
with respect to $f(\cdot)$, which is already an improvement over the result of \cite{KoutsoupiasP13}.
However, we can actually derive the optimal ratio using the same approach as that for $\FT$, which yields an even better competitive ratio for the online auction problem.

\begin{theorem}\label{them-ratio}
The optimal competitive ratio of (offline) digital good auction with respect to the benchmark
$f(\bids)=\max(4\bidi[2],3\bidi[3],4\bidi[4],\dots,n\bidi[n])$ is at most $4.12$.\footnote{the actual ratio for a fixed $n$ is $1-\sum\limits_{i=2}^n\left(\frac{-1}{n}\right)^{i-1} \frac{i}{i-1} \binom{n-1}{i-1} +  \frac{3n}{2 (n-2)}\left( \left(1-\frac{2}{n}\right)^{n-1} + 1- \frac{2}{n}\right).$} 
\end{theorem}

As $f(\bids)=4\bidi[2]$ for $n=2, 3$, and $4$ bidders we get competitive ratio of $4$, which exactly matches the lower bound. Therefore, our online auction is optimal for the case of $2$, $3$, and $4$ bidders.


\subsection{The benchmark $\maxV$ }

The results of \cite{KoutsoupiasP13} carry over for another standard benchmark, namely, the maximum Vickery $\maxV$. As the exact competitive ratio of the optimal offline auction against $\maxV$ was shown in \cite{CGL14} to be $e-1$ and since $2\maxV(\vals)=\FL(\vals)$ for $n=2$ bidders, the approach of \cite{KoutsoupiasP13} implies the following claim.

\begin{theorem}[Koutsoupias et al.\cite{KoutsoupiasP13}]
The competitive ratio of the online sampling auction of \cite{KoutsoupiasP13} is at most $2(e-1)$ against the $\maxV$ benchmark. The competitive ratio of any online auction against $\maxV$ is at least $2$.
\end{theorem}

Interestingly, if we run $\Mon$ exactly as proposed in \cite{KoutsoupiasP13}, i.e., as a sequence of $\Moff$ tailored to $\FT$, then $\Mon$ appears to be specifically well suited for the $\maxV$ benchmark. This observation once again highlights how useful is the idea of thinking about the problem with respect to different benchmarks.

\begin{theorem}
Let $\Moff$ be a $\beta$-competitive auction against the $\FT$ benchmark.
The online sampling auctions composed of a sequence of offline auctions $\Moff$ against $\FT$ is $\beta$-competitive against $\maxV$.
\end{theorem}
\begin{proof}
Similar to the proof of Theorem~\ref{th:plustwo}, we proceed by induction on the number of
bidders. We have
\[
\beta\cdot \Mon[k-1](\bids^{[k]})\ge\frac{1}{k}\sum_{i=1}^{k}\maxV(\bidsmi^{[k]}).
\]
We sort the bids in $\bids^{[k]}:\bidi[1]\ge\dots\ge\bidi[k]$.
For a fixed $\ell$, we want to estimate how the revenue of $\Mon[k-1]$ is compared to
$\bidi[\ell]$. For each $i>\ell$ we have $\FT(\bidsmi^{[k]})\ge(\ell-1)\cdot\bidi[\ell]$; and for $i\le\ell$, we have $\FT(\bidsmi^{[k]})\ge(\ell-2)\cdot\bidi[\ell]$. Hence,
\begin{eqnarray*}
\frac{1}{k}\sum_{i=1}^{k}\FT(\bidsmi^{[k]}) &\ge& \max\InParentheses{\frac{k-2}{k}\bidi[2],
\dots,\frac{(\ell-1)(k-\ell)+(\ell-2)\ell}{k}\bidi[\ell],\dots,\frac{(k-2)k}{k}\bidi[k]}\\
       &=& \max\InParentheses{\frac{k-2}{k}\bidi[2],\dots, \frac{(\ell-1)k-\ell}{k}\bidi[\ell],\dots,\frac{(k-1)k-k}{k}\bidi[k]}\\
       &\ge& \maxV(\bids)-\frac{1}{k}\FT(\bids).
\end{eqnarray*}
Thus, the online sampling auction by running a $\beta$-competitive auction against $\FT$ benchmark is $\beta$-competitive against $\maxV$.
\end{proof}
\begin{corollary}
The competitive ratio of online auctions against $\maxV$ is between $2$ and $2.42$.
\end{corollary}


\bibliographystyle{plain}
\bibliography{bibs,game}

\newpage
\appendix

\section{Proof of Theorem~\ref{them-ratio}}

\begin{proof}
By the same argument as in \cite{CGL14} for $\FT$, the matching lower bound for the optimal competitive ratio
is achieved by the equal revenue distribution with the support $\mathbb{R}^n_{\ge 1}$.  For $n\leq 4$,
$f(\bids)=4\bidi[2]$. In the following, we always assume $n>4$.
We first observe that $f(\bids)=\max(4\bidi[2], \EFL(\bids))$.

We recall that equal revenue distribution $\dist^n$ over the bid vectors is i.i.d. with the density function $\weight(b)=\frac{1}{b^2}$ and cumulative density $1-\frac{1}{b}$ supported on $[1,\infty)$. Let $\bn$ be a random vector drawn from $\dist^n$. 
The key technical problem for us is to compute the expected value of the benchmark $f(\bn)$.
Following~\cite{GoldbergHKS04}, we compute the probability $\prob{f(\bn)\geq z}$ for any given $z$. Since $f(\bn)$ is at least $n$, we may only consider 
$z\ge n$. Let a random variable $V_i$ be the $i$-th largest bid in $\bn$. We also define a set of random variables 
\[F_{n,k}=\max_{i=1,2,\ldots, n} (k+i)\cdot V_i.\]
Let ${\cal H}_i$ denote the event
\[V_i \geq \frac{z}{k+i} \mbox{ \ and \ } \bigwedge_{j=i+1, i+2, \ldots, n} V_j <   \frac{z}{k+j}. \]
The probability of ${\cal H}_i$ can be written as
\[\prob{{\cal H}_i}= \binom{n}{i} \left(\frac{k+i}{z}\right)^i \prob{F_{n-i, k+i}<z}.\]

Since ${\cal H}_i$'s are mutually exclusive and the event $F_{n,k}\geq z$ is the union of  ${\cal H}_i$ for $i=1,2,\ldots, n$, we get
\begin{equation}\label{equ:F-n-k}
\prob{F_{n,k}\geq z} = \sum_i \prob{{\cal H}_i} = \sum_i \binom{n}{i} \left(\frac{k+i}{z}\right)^i \prob{F_{n-i, k+i}<z}.
\end{equation}
This gives a recursive relation for $\prob{F_{n,k}\geq z} $ and the boundary condition is $\prob{F_{0,k}\geq z}=0 $.
This recursion has been solved in~\cite{GoldbergHKS04}:
\[\prob{F_{n,k}\geq z} = 1- \left(\frac{z-k}{z}\right)^n \left(\frac{z-k-n}{z-k}\right). \]

Let ${\cal H}'_2$ denote the event $V_2 \geq \frac{z}{4}$ and  $\bigwedge_{j=3, 4, \ldots, n} V_j <   \frac{z}{j}$.
Then
\[\prob{{\cal H}'_2}=\binom{n}{2} \left(\frac{4}{z}\right)^2 \prob{F_{n-2, 2}<z}.\]
This implies that
\begin{align*}
&\ \ \prob{f(\bn)\geq z} \\
&= \prob{{\cal H}'_2} +\sum_{j=3, 4, \ldots, n} \prob{{\cal H}_j}\\
&= \prob{{\cal H}'_2} +\prob{F_{n,0}\geq z}-\prob{{\cal H}_1}- \prob{{\cal H}_2}\\
&= \binom{n}{2} \left(\frac{4}{z}\right)^2 \prob{F_{n-2, 2}<z} +\frac{n}{z}-\frac{n}{z} \prob{F_{n-1, 1}<z}- \binom{n}{2} \left(\frac{2}{z}\right)^2 \prob{F_{n-2, 2}<z}\\
&=\frac{n}{z}- \frac{n}{z} \left(\frac{z-1}{z}\right)^{n-1} \left(\frac{z-n}{z-1}\right) + \frac{6 n (n-1)}{z^2} \left(\frac{z-2}{z}\right)^{n-2} \left(\frac{z-n}{z-2}\right)
\end{align*}
Therefore, we have
\begin{align*}
\expect{f(\bn)} &= \int_{0}^{\infty} \prob{f(\bn)\geq z} dz  \\
&= n+  \int_{n}^{\infty}\left(\frac{n}{z}- \frac{n}{z} \left(\frac{z-1}{z}\right)^{n-1} \left(\frac{z-n}{z-1}\right) + \frac{6 n (n-1)}{z^2} \left(\frac{z-2}{z}\right)^{n-2} \left(\frac{z-n}{z-2}\right) \right) dz\\
&= n-n\sum_{i=2}^n\left(\frac{-1}{n}\right)^{i-1} \frac{i}{i-1} \binom{n-1}{i-1} +  6 n (n-1) \int_{n}^{\infty}\frac{1}{z^2} \left(\frac{z-2}{z}\right)^{n-2} \left(\frac{z-n}{z-2}\right) dz.\\
\end{align*}

The integration part is
\begin{align*}
&\ \ \int_{n}^{\infty}\frac{1}{z^2} \left(\frac{z-2}{z}\right)^{n-2} \left(\frac{z-n}{z-2}\right) dz\\
&=\int_{n}^{\infty}\frac{1}{z^2} (1-\frac{n}{z}) \left(\frac{z-2}{z}\right)^{n-3} dz\\
&=\int_{n}^{\infty}\frac{1}{z^2} (1-\frac{n}{z}) \left( \sum_{i=0}^{n-3} \binom{n-3}{i} \frac{(-2)^i}{z^i}  \right) dz\\
&=\sum_{i=0}^{n-3} \binom{n-3}{i} (-2)^i  \int_{n}^{\infty} \left(\frac{1}{z^{i+2}}-\frac{n}{z^{i+3}}\right) dz\\
&=\sum_{i=0}^{n-3} \binom{n-3}{i} (-2)^i  \left(\frac{1}{(i+1) n^{i+1}}-\frac{1}{(i+2) n^{i+1}}\right) \\
&=\sum_{i=0}^{n-3} \binom{n-3}{i} (-2)^i  \frac{1}{(i+1)(i+2) n^{i+1}}\\
&=\frac{n}{4 (n-1)(n-2)} \sum_{i=0}^{n-3}  \binom{n-1}{i+2} \left(\frac{-2}{n}\right)^{i+2}  \\
&=\frac{n}{4 (n-1)(n-2)}\left( \left(1-\frac{2}{n}\right)^{n-1} - 1- \binom{n-1}{1}\frac{-2}{n}\right)\\
&=\frac{n}{4 (n-1)(n-2)}\left( \left(1-\frac{2}{n}\right)^{n-1} + 1- \frac{2}{n}\right)
\end{align*}
Therefore, we have
\[ \expect{f(\bn)}=n-n\sum_{i=2}^n\left(\frac{-1}{n}\right)^{i-1} \frac{i}{i-1} \binom{n-1}{i-1} +  \frac{3n^2}{2 (n-2)}\left( \left(1-\frac{2}{n}\right)^{n-1} + 1- \frac{2}{n}\right). \]
And the competitive ratio is
\[1-\sum_{i=2}^n\left(\frac{-1}{n}\right)^{i-1} \frac{i}{i-1} \binom{n-1}{i-1} +  \frac{3n}{2 (n-2)}\left( \left(1-\frac{2}{n}\right)^{n-1} + 1- \frac{2}{n}\right). \]
\end{proof}

\end{document}